\documentclass[twocolumn,journal]{IEEEtran}

\usepackage{amsmath, mathrsfs, mathtools}
\usepackage{amsfonts}
\usepackage{amsthm}
\usepackage{amssymb}
\usepackage{graphicx}
\usepackage{xcolor}
\usepackage{multirow}
\usepackage{graphicx}
\usepackage{stmaryrd}
\usepackage[backend=bibtex,giveninits=true,sorting=none]{biblatex}

\renewcommand*{\bibfont}{\footnotesize}

\allowdisplaybreaks

\allowdisplaybreaks

\newcommand{\myhash}{%
  {\settoheight{\dimen0}{C}\kern-.05em\, \resizebox{!}{\dimen0}{\raisebox{\depth}{\#}}}}

\usepackage{caption}

\usepackage{subcaption}

\usepackage{algorithm}
\usepackage{algpseudocode}
\usepackage{pifont}
\usepackage{varwidth}

\newcommand{\herm}{{\sf T}}

\usepackage{multicol}


\newtheorem{theorem}{\bf Theorem}

\newtheorem{corollary}{\bf Corollary}

\theoremstyle{remark}
\newtheorem{remark}{\bf Remark}
\newtheorem{example}{\bf Example}
\theoremstyle{definition}
\newtheorem{definition}{\bf Definition}

\usepackage{float}
\usepackage{afterpage}
\usepackage{balance}

\long\def\comment#1{}

\newcommand{\ben}{\begin{enumerate}}
\newcommand{\een}{\end{enumerate}}

\newcommand{\beq}{\begin{equation}}
\newcommand{\eeq}{\end{equation}}

\newcommand{\bi}{\begin{itemize}}
\newcommand{\ei}{\end{itemize}}

\newcommand{\PP}{\mathbb{P}}
\newcommand{\RR}{\mathbb{R}}

\newcommand{\EE}{\mathbb{E}}

\newcommand{\av}{{\bf a}}
\newcommand{\bv}{{\bf b}}

\newcommand{\xv}{{\bf x}}

\newcommand{\Am}{{\bf A}}
\newcommand{\Bm}{{\bf B}}

\newcommand{\Id}{{\bf I}}

\renewcommand{\vec}{{\rm vec}}

\title{On the Restricted Isometry Property of Centered Self Khatri-Rao Products}
\author{Alexander Fengler and Peter Jung\\
  Communications and Information Theory Group,
  Technische Universit\"{a}t Berlin\\[.2em]
  {\em \{fengler,peter.jung\}@tu-berlin.de}}

\begin{document}

\maketitle

\begin{abstract}
  In this work we establish the Restricted Isometry Property (RIP) of
  the centered column-wise self Khatri-Rao (KR) products of
  $n\times N$ matrix with iid columns drawn either uniformly from a
  sphere or with iid sub-Gaussian entries. The self KR product is an
  $n^2\times N$-matrix which contains as columns the vectorized (self)
  outer products of the columns of the original $n\times N$-matrix.
  Based on a result of Adamczak et al.  we show that such a centered self
  KR product with independent heavy tailed columns has small RIP
  constants of order $s$ with probability at least $1-C\exp(-cn)$
  provided that $s\lesssim n^2/\log^2(eN/n^2)$.  Our result is
  applicable in various works on covariance matching like in activity
  detection and MIMO gain-estimation.
\end{abstract}

\newcommand{\dimParam}{N}
\newcommand{\dimPilots}{n}

\section{Introduction}
In estimation and recovery problems related to empirical second
moments, e.g. covariance matching, one often observes a matrix, which is
a noisy linear combination of outer products
$\{\av_i\av_i^\herm\}_{i=1}^N$ of $N$ random but known  vectors
$\av_i\in\RR^n$. The goal is then to estimate the unknown coefficients from this
observed $n\times n$ matrix. In the prototypical example this is an
empirical covariance matrix. This problem appears in matrix and tensor
recovery problems and many recent applications, see
e.g. \cite{Romera:2016, Dasarathy:2015, Ma:2010,Duan:2016,Haghighatshoar:MVV:arxiv:18} and massive MIMO
\cite{Fengler:MassiveAccess:arxiv:2019}. In the case of sparse linear
combinations this yields a compressed sensing
\cite{Candes2005,Donoho2006a} problem with a random but structured
$n^2\times N$ measurement matrix.

There are several known properties
which ensure robust and stable $\ell_p$-recovery guarantees of the
vector of unknown but sparse (or compressible) coefficients. Among
them is the restricted isometry property (RIP) of order $s$ which
ensures that a $m\times N$ measurement matrix maps $s$-sparse vectors
almost-isometrically, i.e., there exists $\delta_s\in[0,1)$ such that
\begin{equation}
  (1-\delta_s)\|\xv\|^2_2\leq\|\Am \xv\|_2^2\leq(1+\delta_s)\|\xv\|_2^2
\end{equation}
holds for all $s$--sparse vectors $\xv\in\RR^N$.  Several upper bounds
on $\delta_{2s}$ have been established to ensure stable and robust
recovery for certain algorithms, see e.g. \cite{candes:rip2008,
  Foucart2013}.  For example, it is known that that $\ell_1$-based
convex recovery algorithms succeed if $\delta_{2s}<\frac{1}{\sqrt{2}}$
\cite{Cai2014}.  For a random matrix with iid. sub-Gaussian components
it is known that this property holds with overwhelming probability for
$m\gtrsim s\cdot\log(N/s)$, see for example \cite{Foucart2013} and
also \cite{dirksen:ripgap} for a further discussions about its
relation to the e.g. the nullspace property.  When imposing additional
structure on a random measurement matrix often more measurements are
required to ensure robust recovery guarantees.  However, for many
random ensembles it has been shown that it is still possible to achieve, up
to log-factors, a linear relation between sparsity $s$ and number of
measurements $m$, meaning that $m\gtrsim s\cdot\text{polylog}(N)$.

In this work we show that the structure imposed by the problem above
indeed also allows robust and stable recovery in the regime
$s\lesssim n^2/\log^2(eN/n^2)$ since $m=n^2$.  This addresses a
conjecture raised in \cite{Khanna:correction} for non-centered KR
product.  Some of the essential proof steps have been sketched already
in \cite{Hag:isit18}.

More precisely, let $\Am = (\av_i)_{i=1,...,N}$ be a random matrix
with independent columns $\av_i \in \RR^{n}$, $\av_i \sim P_\av$.  The
(column-wise) self Khatri-Rao product of the matrix $\Am$ is defined
as
\beq
(\Am \odot \Am)_i := \vec(\av_i \otimes \av_i)=\vec(\av_i\av_i^\herm)
\eeq
where the
matrix $(\av \otimes\bv)_{ij}=(\av\bv^\herm)_{ij}:=a_ib_j$ is the outer product\footnote{
  $\vec: \RR^{n_1\times n_2}\to\RR^{n_1n_2}, \vec(\Bm) = \bv$ with
  $\bv_{i+jn_2} = \Bm_{ij}$ is the vectorization operation
  identifying a matrix with a vector.}  of the vector $\av$ and
$\bv$.

We assume the columns $\av_i$ to be normalized in expectation such
that $\EE\{\|\av_i\|^2\} = n$ and are drawn from an isotropic
distribution, i.e.  \beq \EE\{\av_i \otimes \av_i\} = \Id_n.  \eeq
First results\footnote{Note that the results have been corrected in
  v3 of the preprint} on the RIP property for self KR products have
been established in \cite{Khanna:KRRIP:arxivv3,Khanna:KRRIP}, \cite{Khanna:correction}. In this work it has
been shown that for $n\gtrsim s\log(N)$ (meaning that
$m=n^2\gtrsim s^2\log^2(N)$) the $n^2\times N$-dimensional KR product
$\Am\odot\Am$ of a centered iid. sub-Gaussian $n\times N$ matrix $\Am$
has RIP with high probability, see \cite[Theorem
3]{Khanna:KRRIP:arxivv3}. Thus, the number of measurements $m$ scales
{\em quadratically} in the sparsity $s$. However, we will show below
that the scaling is indeed {\em linear} when centering the
KR product.

We will use the work of \cite{Ada2011} to prove a bound on the RIP
constant of the centered and normalized KR product $\mathcal{A} \in \RR^{n^2\times N}$:
\beq
\label{eq:centered_kr}
\mathcal{A}_i := \kappa(n)\vec (\av_i \otimes \av_i - \Id_n)
\eeq
It is easy to see that $\EE\{\mathcal{A}_i\} = 0$. 
$\kappa(n)$ is a normalization factor to ensure that the columns of $\mathcal{A}$ are still
normalized after centering:
$\EE\{\|\mathcal{A}_i\|^2\} = n^2$.
In general
\beq
\kappa(n) = \frac{n^2}{\EE\{\|\vec (\av_i \otimes \av_i - \Id_n)\|^2_2\}}.
\label{eq:kappa}
\eeq
Note that for a vector $\av=(a_1,\dots,a_N)$ one has
\begin{equation}
  \begin{split}
    &\|\vec (\av \otimes \av - \Id_n)\|^2_2
    = 
    \sum_{i,j=1}^n (a_ia_j - \delta_{ij})^2 \\
    &= \sum_{i\neq j}^n a_i^2a_j^2 + \sum_{i=1}^n (a_i^2 - 1)^2 
    = \left(\sum_{i=1}^n a_i^2\right)^2 - 2\sum_{i=1}^n a_i^2 + n
    \label{eq:kappa1}
  \end{split}
\end{equation}
\begin{example}[iid with normalized 2.~moment]
  \label{ex:sub_gauss}
  Let $\av \in \RR^n$ be a random vector with components
  $\{a_i\}_{i=1}^n$ being independent copies of $a\sim P_a$ with
  $\EE\{a^2\} = 1$.  Then  
  \beq
  \EE\{\|\vec (\av \otimes \av - \Id_n)\|^2_2\} 
  = n(n-2+\EE\{a^4\})
  \eeq
\end{example}
\begin{example}[Constant Amplitude]
  Let $\av \in \RR^n$ be a random vector such that it's components
  $\{a_i\}_{i=1}^n$ are independent copies of a Rademacher random
  variable $a$, i.e., with fixed amplitude $a^2 = 1$ and uniformly
  distributed sign.
  From the previous
  example it follows that \beq \kappa(n)= \frac{n^2}{n(n-1)}=
  \frac{n}{n-1}.  \eeq
\end{example}
\begin{example}[Spherical Distribution]
    \label{ex:spherical}
    Let $\av \in \RR^n$ be drawn uniformly from the sphere with radius $\sqrt{n}$.
    Then $\|\av\|_2^2=\sum_i a_i^2 = n$ and it can easily be checked
    that \eqref{eq:kappa1} gives:
    \beq
        \kappa(n)= \frac{n}{n-1}.
    \eeq
\end{example}

\section{RIP for Centered KR Products}
The $\psi_\alpha$-norm for $\alpha\geq 1$ of a real-valued random
variable $Y$ can\footnote{this definitions are not unique in the
  literature such that these norms may differ in constants}
formally be defined as:
\begin{equation}
  \begin{split}
    \|Y\|_{\psi_\alpha}
    &=\inf\{ K>0: \mathbb{E}\,\exp(|Y|^\alpha/K^\alpha)\leq 2\}
  \end{split}
  \label{eq:def:psinorm}
\end{equation}
Note that $\|Y\|_{\psi_\alpha}\leq \|Y\|_{\psi_\beta}$ for
$\alpha\geq\beta\geq 1$. If $\|Y\|_{\psi_\alpha}<\infty$ is satisfied,
the random variable $Y$ is called {\em sub-Gaussian} for $\alpha=2$ and {\em
  sub-exponential} for $\alpha=1$.
The definitions above extend in
a canonical way to random vectors.  The $\psi_\alpha$-norm of a
random vector $X$ is defined as the best uniform bound on the
$\psi_\alpha$-norm of its marginals:
\beq \|X\|_{\psi_\alpha} :=\sup_{\|x\|_2 = 1} \|\langle
X,x\rangle\|_{\psi_\alpha}.  \eeq
$\psi_\alpha$ random variables and
vectors for $\alpha<2$ are often called heavy tailed.
Note that this terminology is also important if the $\psi_2$-norm of a random
vectors grows with its dimension.

\subsection{RIP for Independent Heavy-tailed Columns}
As introduced above, the KR product of a random matrix with
independent sub-Gaussian isotropic columns is itself a matrix with
heavy tailed (sub-exponential) independent columns having a special
structure.  The RIP properties for the column-independent model with
normalized sub-Gaussian isotropic columns have been established in
\cite{vershynin_2018}.
In a series of works
\cite{Ada2011,Guedon2014:heavy:columns,Guedon2017} the heavy tailed
column independent model has been further investigated and concrete
results can be found for various ensembles.  However, the previously
investigated ensembles not explicitly discuss the structure
imposed by KR products.
Thus, we make use of the following generic RIP result from
\cite[Theorem~3.3]{Ada2011} for matrices with iid sub-exponential
columns:
\begin{theorem}[Theorem~3.3 in \cite{Ada2011}]
\label{thm:adam}
Let $m\geq 1$ and $s,N$ be integers such that $1\leq s\leq \min (N,m)$.
Let $X_1,...,X_N \in \RR^m$ be independent $\psi_1$ random vectors
normalized such that $\EE\{\|X_i\|^2\} = m$ and
let $\psi = \max_{i\leq N}\|X_i\|_{\psi_1}$. Let $\theta^\prime \in(0,1)$,
$K,K^\prime \geq 1$ and set $\xi = \psi K + K^\prime$. Then for the
matrix $A$ with columns $X_i$, $A := (X_1|...|X_N)$
\beq
\delta_s\left(\frac{A}{\sqrt{m}}\right) \leq C \xi^2\sqrt{\frac{s}{m}}
\log\left(\frac{eN}{s\sqrt{\frac{s}{m}}}\right)
+ \theta^\prime
\eeq
holds with probability larger then
\begin{align}
    1 &- \exp\left(-cK\sqrt{s}\log\left(\frac{eN}{s\sqrt{\frac{s}{m}}}\right)\right) \\
&- \PP\left(\max_{i\leq N}\|X_i\|_2\geq K^\prime\sqrt{m}\right) 
- \PP\left(\max_{i\leq N}\left|\frac{\|X_i\|_2^2}{m} - 1\right|\geq \theta^\prime\right),
\label{eq:adam_error}
\end{align}
where $C,c > 0$ are universal constants.
\end{theorem}
We shall use this theorem for $X_i = \mathcal{A}_i$ and $m = n^2$.
The key to get a good bound from Theorem \ref{thm:adam} is to 
\begin{enumerate}
    \item Show the marginals of 
        the columns of $\mathcal{A}$ have sub-exponential tails
        with a sub-exponential norm, which is independent of the dimension $n$.
    \item Show that the norm of the columns of $\mathcal{A}$ concentrate well around
        their mean.
\end{enumerate}
If the columns of $\mathcal{A}$ are exactly normalized, then the second point is trivially
fulfilled, the latter two terms of \eqref{eq:adam_error} vanish
and we can choose $\theta^\prime >0$ and $K^\prime\geq1$ to be arbitrary small.
We can use the following
corollary for matrices with constant norm:
\begin{corollary}
    \label{cor:constant_norm}
    Let all parameters be as in Theorem \ref{thm:adam} with the additional requirement
    that $\|X_i\|_2^2 = m$. Additionally we assume that $m\leq N$.
    Then the RIP constant of order $s$ of $\frac{A}{\sqrt{m}}$
    satisfies
    \beq
    \delta_{s}\left(\frac{A}{\sqrt{m}}\right) < \delta
    \eeq
    with probability at least $1 - \exp(-C^\prime\sqrt{c_{\xi,\delta} m})$
    as long as
    \beq
    s \leq c_{\xi,\delta}\frac{m}{\log^2\left(\frac{eN}{c_{\xi,\delta}m}\right)}.
    \eeq
    Where $c_{\xi,\delta} = \min(1,(\frac{\delta}{C\xi^2})^2)$
    and $C,C^\prime$ are some universal constants.
\end{corollary}
\begin{proof}
  Let us abbreviate $\delta_s=\delta_{s}\left(\frac{A}{\sqrt{m}}\right)$.
  Since $\|X_i\|^2_2 = m$, the last two terms in \eqref{eq:adam_error} vanish for all
  $K^\prime>1$ and $\theta^\prime > 0$.
    \begin{equation}
        \delta_{s}\leq
        C \xi^2 \sqrt{\frac{s}{m}} \log\left(\frac{eN}{s \sqrt{s/m}}\right)
        =: D
        \label{eq:deltas}
    \end{equation}
    with probability larger than 
    \begin{align}
        \mathbb{P}(\delta_{s}\leq D) \geq 1
        &-\exp\left(-cK \sqrt{s} \log \left(\frac{eN}{s \sqrt{s/m}}\right)\right)
    \end{align}
    Let $s \leq cm/\log^2(e\frac{N}{c m})$ for any $0 <c\leq1$.
    Note that the conditions $c\leq 1$ and $N\geq m$ guarantee that
    $\log (e\frac{N}{cm})\geq 1$.
    Plugging into \eqref{eq:deltas}
    we see that the RIP-constant
    satisfies 
    \begin{align}
        \delta_{s}&\leq C\xi^2\sqrt{c}\frac{
            \log(e(\frac{N}{c m})^{3/2}\log^3(e\frac{N}{c m}))
        }{\log(e\frac{N}{c m})} \\
        &= C\xi^2\sqrt{c}
        \left(\frac{3}{2} +
        \frac{3\log\log e\frac{N}{c m}}{\log e\frac{N}{c m}}
        \right)
        \\
        &\leq C\xi^2\sqrt{c}\left(\frac{3}{2} + \frac{3}{e} \right)\\
        &\leq 3C\xi^2\sqrt{c}
    \end{align}
    where in the first line we made use of $m\leq N$ and in the last line we used
    $\log\log x/\log x \leq 1/e$. This bound fails with probability:
    \begin{align}
        \mathbb{P}(\delta_{s}>D)
        &\leq \exp\left(-\hat{c}K\sqrt{s}\log \left(e\frac{N\sqrt{m}}{s^{3/2}}\right)\right) \\
        &\leq \exp\left(-\hat{c}K\sqrt{s}\log \left(e\frac{N}{m}\right)\right) \\
        &\leq \exp(-\hat{c}K\sqrt{c}\sqrt{m})
    \end{align}
    where in the second line it was used that $s \leq m$. The statement of
    the Corollary follows by choosing $c$ small enough such that $\delta_s\leq\delta$.
\end{proof}
\subsection{The Case of Sub-Gaussian iid Columns $\av_i$}
\label{sec:sub-gauss}
We will show here that Corollary \ref{cor:constant_norm} holds almost
unchanged if $X_i = \mathcal{A}_i$ where $\mathcal{A}_i$ are the
columns of the centered self KR product of a matrix $\Am$ with sub-Gaussian
iid entries as defined in \eqref{eq:centered_kr} and Example
\ref{ex:sub_gauss}.  First we need to show, that the columns
$\mathcal{A}_i$ are sub-exponential with a $\psi_1$-norm independent
of $n$. This is a consequence of the Hanson-Wright inequality, which
states that every centered quadratic form of independent sub-Gaussian
random variables is sub-exponential:
\begin{theorem}[Hanson-Wright inequality]
    \label{thm:hanson_wright}
    Let $X = (X_1,...,X_n)\in \RR^n$ be a random vector with independent components $X_i$
    which satisfy $\EE X_i = 0$ and $\|X_i\|_{\psi_2}\leq B$.
    Let $Y$ be a $n\times n$ matrix. Then, for every $t\geq0$,
    \beq\begin{split}
      \PP\{|X^\top YX &- \EE X^\top YX|>t\}\\
      &\leq 2\exp\left(-c\min\left(\frac{t^2}{B^4\|Y\|_F^2},\frac{t}{B^2\|Y\|}\right)\right)
    \end{split}\eeq
\end{theorem}
\begin{proof}
    See \cite{Rud2013}
\end{proof}
With $\|Y\|$ and $\|Y\|_F$ we denote here operator norm and Frobenius
norm of the matrix $Y$.
Note that a RV with such a mixed tail behavior is 
especially sub-exponential. This can be seen by bounding its moments. 
Let $Z=X^\top YX - \EE X^\top YX$ be a RV with
\beq
\PP(|Z|>t)\leq2\exp\left(-c\min\left(\frac{t^2}{B^4\|Y\|_F^2},\frac{t}{B^2\|Y\|}\right)\right)
\eeq
Since $\|Y\| \leq \|Y\|_F$, we have
$\PP(|Z|>t) \leq 2\exp(-c\min(x(t)^2,x(t)))$ for $x(t) = \frac{t}{B^2\|Y\|_F}$.
It follows
\begin{equation}\begin{split}
    &\EE |Z|^p
    = \int_0^\infty \PP(|Z|^p>u)\mathrm{d}u
    = p\int_0^\infty \PP(|Z|>t) t^{p-1}\mathrm{d}t \\
    &\leq 2p(B^2\|Y\|)^p\left(\int_0^1e^{-x^2}x^{p-1}\mathrm{d}x
    + \int_1^\infty e^{-x}x^{p-1} \mathrm{d}x\right) \\
    &\leq 2p(B^2\|Y\|)^p\left(\Gamma(p/2) + \Gamma(p)\right) \\
    &\leq 4p(B^2\|Y\|)^p\Gamma(p)\leq 4p(pB^2\|Y\|)^p
  \end{split}\end{equation}
where $\Gamma(\cdot)$ is the Gamma function. So
\beq
(\EE|Z|^p)^\frac{1}{p} \leq cpB^2\|Y\|
\label{eq:subexp:momentbound}
\eeq
which is equivalent to $\|Z\|_{\psi_1} \leq cB^2\|Y\|$
by elementary properties of sub-exponential
random variables.
\begin{theorem}
    \label{thm:kr_sub_exp}
    Let $\Am=\left(A_{ij}\right)$ be a random matrix with sub-Gaussian
    iid entries, satisfying $\|A_{ij}\|_{\psi_2} \leq B$ and
    $\EE A_{ij} = 0$ and normalized such that $\EE\{A_{ij}^2\} = 1$.
    Let $\mathcal{A}_i$ be the $i$'th column of the corresponding
    centered self KR product as defined in \eqref{eq:centered_kr} and
    Example \ref{ex:sub_gauss}.  Then
    \beq
    \|\mathcal{A}_i\|_{\psi_1}
    = \sup_{\|y\|_2=1}\|\langle\mathcal{A}_i,y\rangle\|_{\psi_1} \leq cB^2
    \eeq
    for some absolute constant $c>0$.
\end{theorem}
\begin{proof}
    Note that we can rewrite
    \begin{align}
    \langle\mathcal{A}_i,y\rangle &=
    \sum_{j,k} \kappa(n) (A_{ij}A_{ik}Y_{jk} - \EE\{A_{ij}A_{ik}\}Y_{jk}) \\
    &= \kappa(n)(\av_i^\top Y\av_i - \EE \av_i^\top Y \av_i)
    \end{align}
    where the matrix $Y$ is chosen such that $\vec(Y) = y$,
    and therefore $\|Y\|_F = \|y\|_2 = 1$. 
    With this, it follows immediately from Theorem \ref{thm:hanson_wright} that
    $\langle\mathcal{A}_i,y\rangle$ is sub-exponential with
    \beq
    \|\langle\mathcal{A}_i,y\rangle\|_{\psi_1} \leq c\kappa(n)B^2 \|Y\|
    \eeq
    for some absolute constant $c>0$.
    It holds that
    $\|Y\| \leq \|Y\|_F = 1$. We see from example \ref{ex:sub_gauss}
    that $\kappa(n) = \frac{n}{n-2+\EE\{A_{ij}^4\}}$. By Jensen inequality
    $\EE\{A_{ij}^4\}\geq\EE\{A_{ij}^2\}^2 = 1$, so $\kappa(n)\leq \frac{n}{n-1} \leq 2$. 
\end{proof}
To apply Theorem \ref{thm:adam} to $\mathcal{A}$ we need to show that the norm of it's
columns concentrate well around their mean. This is the subject of the following theorem.
\begin{theorem}
    \label{thm:deviation}
    Let $\{\mathcal{A}_i\}_{i=1}^N$ be the columns of the centered
    self KR product of a centered, normalized sub-Gaussian iid matrix $\Am$ as in Theorem
    \ref{thm:kr_sub_exp}. Let $P_a$ denote the distribution of the
    entries of $\Am$.  Then for $i=1,\dots, N$ it holds:
    \beq
        \PP\left(\max_{i\leq N}\left|\frac{\|\mathcal{A}_i\|^2_2}{n^2} -
            1\right|\geq t\right)
            \leq C\exp\left(\log N-\frac{c}{B^2}\sqrt{t}n\right)
    \eeq
    if
    $n$ satisfies
    \beq
    n\geq1 + (\EE a^4 - 1) (3/t-1)
    \eeq
\end{theorem}
\begin{proof}
    By union bound we have that 
    \beq
    \PP\left(\max_{i\leq N}\left|\frac{\|\mathcal{A}_i\|^2_2}{n^2} - 1\right|\geq t\right)
    \leq N
\PP\left(\left|\frac{\|\mathcal{A}_i\|^2_2}{n^2} - 1\right|\geq t\right)
    \eeq
    Furthermore, with the abbreviation $S := \sum_{i=1}^n a_i^2$,
    we have (see example \ref{ex:sub_gauss})
    \beq
    \|\mathcal{A}_i\|^2_2 = \kappa(n)(S^2 - 2S + n)
    \eeq
    which can be rewritten as
    \beq
    \|\mathcal{A}_i\|^2_2 = \kappa(n)((S-n)^2 + 2(n-1)(S-n) + n(n-1)).
    \eeq
    Example \ref{ex:sub_gauss} shows that
    $n^2/\kappa(n) = n(n-2+\EE a^4)$, thus 
    \begin{align}
        \frac{\|\mathcal{A}_i\|^2_2}{n^2} - 1 
        &=\kappa(n)\frac{(S-n)^2 + 2(n-1)(S-n) - n(\EE a^4 - 1))}{n^2} \\
        &=: a + b + c
    \end{align}
    with $a:= \frac{\kappa(n)(S-n)^2}{n^2}$,
    $b:=\frac{2\kappa(n)(n-1)(S-n)}{n^2}$ and
    $c:= \frac{\EE a^4 - 1}{n-2+\EE a^4}$.
    We can estimate the one sided tail $\PP(a+b+c>t)$ by
    \begin{align}
        \PP(a + b + c > t) 
        &\leq \PP\left(a>\frac{t}{3}\right)
        + \PP\left(b>\frac{t}{3}\right)
        + \PP\left(c>\frac{t}{3}\right)
    \end{align}
    Therefore $S-n = \sum_{i=1}^n(a_i^2 - 1)$ is a sum of independent zero
    mean sub-exponential random variables with $\|a_i^2 - 1\|_{\psi_1}
    \leq cB^2$, as a centering argument
    and the identity $\|X^2\|_{\psi_1} = \|X\|^2_{\psi_2}$ for sub-Gaussian random variables
    $X$ shows,
    e.g. \cite[Ch. 2.7]{vershynin_2018}.
    Therefore the elemental 
    Bernstein inequality gives that
    \beq
    \PP(|S-n|>nt) \leq 2\exp\left(-\tilde{c}n\min\left(\frac{t^2}{B^4},\frac{t}{B^2}\right)\right)
    \eeq
    Then the same argument as in \eqref{eq:subexp:momentbound} shows that
    \beq
    \PP(|S-n|>nt) \leq 2\exp(-cnt/B^2)
    \eeq
    for some constant $c$. So in particular
    \begin{align}
        \PP\left(b>\frac{t}{3}\right) &\leq 2\exp\left(-ct\frac{n^2}{6B^2\kappa(n)(n-1)}\right)\\
                                      &\leq 2\exp\left(-\frac{c^\prime}{B^2} tn\right)
    \end{align}
    where in the last step we used that $n^2/\kappa(n) \leq n^2$
    and $n/(n-1) \leq 2$.
    The probability of deviation of $a$ can be bound as follows:
    \begin{align}
    \PP\left(a>\frac{t}{3}\right)
    &= \PP\left(|S-n|>\sqrt{\frac{n^2t}{3\kappa(n)}}\right) \\
    &\leq 2\exp\left(-\frac{c}{B^2}\sqrt{t\frac{n^2}{3\kappa(n)}}\right)\\
    &\leq 2\exp(-c^\prime \sqrt{t}n/B^2)
    \end{align}
    Finally
    \begin{align}
        \PP\left(c>\frac{t}{3}\right)
        = \begin{cases}
            0 \text{ if } n\geq1 + (\EE a^4 - 1) (3/t-1)\\
            1 \text{ o.w.}
        \end{cases}
    \end{align}
    For the other tail $\PP(a+b+c < -t) = \PP(-a - b - c>t)$, notice
    that $a$ and $c$ are non-negative. For $a$ this is obvious, for $c$ it follows
    from Jensen inequality and $\EE a^2 = 1$. Therefore
    $\PP(-a-b-c>t) \leq \PP(-b>t) = \PP(b>t) \leq \PP(a+b+c >t)$.
\end{proof}
The third term in \eqref{eq:adam_error} is the probability of a one sided deviation
of $\|X_i\|$ and therefore we can bound it by the same term as in theorem \ref{thm:deviation}:
\begin{align}
    \label{eq:deviation_bound_2}
\PP\left(\max_{i\leq N}\|X_i\|_2\geq K^\prime n\right)
&\leq N\PP\left(\frac{\|X_i\|_2}{n^2} -1 >K^{\prime2} - 1\right) \nonumber\\
&\leq N\PP\left(\left|\frac{\|X_i\|_2}{n^2} -1\right| >K^{\prime2} - 1\right)
\end{align}
Now we can state that the result of Corollary \ref{cor:constant_norm}
holds almost unchanged, except for different constants, 
for the self KR product of an iid sub-Gaussian matrix:
\begin{theorem}
    \label{thm:sub_gauss}
    Let $n\geq 1$ and $s,N$ be integers such that $n^2 \leq N$
    and $1\leq s\leq n^2$.
    Let $\Am \in \RR^{n\times N}$ be a random matrix with sub-Gaussian iid
    entries, distributed according to $P_\av$, with $\EE a = 0$,$\EE a^2 = 1$
    and $\|a\|_{\psi_2} \leq B$.
    Let $\mathcal{A} \in \RR^{n^2\times N}$ be the centered and rescaled
    self-KR product of $\Am$ as defined in \eqref{eq:centered_kr}.
    Then the RIP constant of order $s$ of $\frac{\mathcal{A}}{n}$
    satisfies
    \beq
    \delta_{s}\left(\frac{\mathcal{A}}{n}\right) < \delta
    \eeq
    for any $\delta>0$ with probability larger then
    \beq
    \PP(\delta_{s}\geq\delta) \geq 1 - C\exp(-cn/B^2)
    \eeq
    as long as
    \beq
    s \leq c_{\xi,\delta}\frac{n^2}{\log^2(\frac{eN}{c_{\xi,\delta}n^2})}
    \eeq
    and
    \beq
        \label{eq:log_condition}
        n \geq \max(c_1\log N,1 + c_2B^2(6/\delta - 1)).
    \eeq
    Where $c_{\xi,\delta} = \min(1,(\frac{\delta}{C^\prime\xi^2})^2)$ with $\xi = c^\prime B^2 + 1$.
    For some universal constants $c,c^\prime,C,C^\prime,c_1,c_2>0$.
\end{theorem}
\begin{proof}
    Theorem \ref{thm:kr_sub_exp} shows that the columns of
    $\mathcal{A}$ are sub-exponential.
    So the prerequisites of Theorem \ref{thm:adam} are fulfilled with
    $\psi = c^\prime B^2$, for some absolute constant $c^\prime >0$,
    and $m = n^2$. We set $\theta^\prime = \delta/2$
    and $K^\prime = \sqrt{1+\theta^\prime}$. Furthermore we can set $K=1$,
    such that
    $\xi = \psi K + K^\prime = c^\prime B^2 + 1$.
    Theorem \ref{thm:deviation},
    with $t = \theta^\prime$ and
    \eqref{eq:deviation_bound_2} show that there exist constants $C,\tilde{c}>0$
    \begin{equation}
      \begin{split}
        &\PP\left(\max_{i\leq N}\|X_i\|_2\geq K^\prime n\right) + 
        \PP\left(\max_{i\leq N}\left|\frac{\|\mathcal{A}_i\|^2_2}{n^2} - 1\right|\geq \theta^\prime\right)\\
        &\leq C\exp\left(\log N - \frac{\tilde{c}}{B^2}n\right)
      \end{split}
    \end{equation}
    if $n \geq 1 + (\EE a^4 - 1)(3/\theta^\prime -1)$. 
    (the latter is simply a constant, since $\EE a^4$ is bounded
    by $c\|a\|_{\psi_2}$ for sub-Gaussian $a$.)
    Choosing $c_1$ in the condition \eqref{eq:log_condition} large enough,
    such that $1/c_1 < \frac{\tilde{c}}{B^2}$, we can guarantee that 
    \beq
    \exp(\log N - \tilde{c}n/B^2)
    \leq \exp(-cn/B^2).
    \eeq
    with $c>0$. So Theorem \ref{thm:adam} gives that
    \beq
    \delta_s\left(\frac{\mathcal{A}}{n}\right) \leq C\xi^2\sqrt{\frac{s}{n^2}}
    \log\left(\frac{eN}{s\sqrt{\frac{s}{n^2}}}\right)
    \eeq
    holds with probability larger then
    \begin{align}
        1 &- C\exp\left(-c\sqrt{s}\log\left(\frac{eN}{s\sqrt{\frac{s}{n^2}}}\right)\right) 
          - \tilde{C}\exp(-\tilde{c}n/B^2)
    \end{align}
    Then, the same calculation as in the proof of Corollary \ref{cor:constant_norm}
    shows, that there is a constant $c_{\xi,\delta}> 0$ such that setting
    $s = c_{\xi,\delta} n^2/\log^2(e\frac{N}{c_{\xi,\delta} n^2})$ leads to the result of this theorem.
\end{proof}
\subsection{Spherical Columns $\av_i$}
Let $\Am\in\RR^{n\times N}$
be a matrix such that its columns $\av_i$ are drawn iid from a sphere
with radius $\sqrt{n}$. See Example \ref{ex:spherical}.
Since the columns are now exactly normalized
we can apply Corollary \ref{cor:constant_norm}, if
we can show that columns of the centered self-KR product $\mathcal{A}$
have sub-exponential marginals, with a sub-exponential norm
independent of the dimension.
For this we can use the following result from \cite{Ada2015} which states
that a random vector which satisfies the convex concentration property
also satisfies the Hanson-Wright inequality:
\begin{theorem}[Theorem 2.5 in \cite{Ada2015}]
    \label{thm:hanson-wright-convex}
    Let $X$ be a mean zero random vector in $\RR^n$, which satisfies the convex concentration
    property with constant $B$, then for any $n\times n$ matrix $Y$ and every $t>0$,
    \begin{align}
        \PP\{|X^\top YX &- \EE X^\top YX|>t\}\nonumber\\
                        &\leq 2 \exp\left(-c\min\left(\frac{t^2}{2B^4\|Y\|_F^2},\frac{t}{B^2\|Y\|}\right)\right)
    \end{align}
\end{theorem}
The convex concentration property is defined as follows
\begin{definition}[Convex Concentration Property]
    Let $X$ be a random vector in $\RR^n$. $X$ has the convex concentration property
    with constant $K$ if for every 1-Lipschitz convex function $\phi: \RR^n\to R$,
    we have $\EE|\phi(X)|<\infty$ and for every $t>0$,
    \beq
    \PP\{|\phi(x) - \EE \phi(X)|\geq t\} \leq 2\exp(-t^2/K^2)
    \eeq
\end{definition}
A classical result states that a spherical random variable
$X \sim \text{Unif}(\sqrt{n}S^{n-1})$
has the even stronger (non-convex) concentration property
(e.g. \cite[Theorem 5.1.4]{vershynin_2018}):
\begin{theorem}[Concentration on the Sphere]
    Let $X\sim\textup{Unif}(\sqrt{n}S^{n-1})$ be uniformly distributed
    on the Euclidean sphere of radius $\sqrt{n}$. Then there is an absolute
    constant $c>0$, such that for every 1-Lipschitz
    function $f:\sqrt{n}S^{n-1}\to\RR$
    \beq
    \PP\{f(X) - \EE f(X)\} \leq 2\exp(-ct^2)
    \eeq
\end{theorem}
So in particular $X$ has the convex concentration property with constant $c$ and
it follows by Theorem $\ref{thm:hanson-wright-convex}$
that it also satisfies the tail bound of the Hanson-Wright inequality.
As shown in \eqref{eq:subexp:momentbound},
this implies that the columns of $\mathcal{A}$ are sub-exponential
with $\|\mathcal{A}_i\|_{\psi_1} \leq C$ for some absolute constant $C>0$.
With this we can apply Corollary \ref{cor:constant_norm}.

\begin{remark}
    In this section we did not specifically use the property that the columns
    of $\Am$ are drawn iid from the sphere, but only their convex concentration property.
    So the results also hold for the larger class of normalized columns with dependent
    entries, i.e. those
    which satisfy
    the convex concentration property. E.g. it is known that
    $X = (x_1,...,x_n)$ satisfies the convex concentration property 
    if its entries are drawn iid without replacement from some fixed
    set of numbers $\{b_1,...,b_m\}$ with $b_i \in [0,1]$. For more examples
    see \cite{Ada2015}. Also note that the sub-Gaussian iid case of
    section \ref{sec:sub-gauss} is not covered by Theorem
    \ref{thm:hanson-wright-convex}, since $X = (x_1,...,x_n)$
    with sub-Gaussian iid $x_i$ does not, in general, have the convex concentration
    property with a constant independent of dimension \cite{Ada2015}.
\end{remark}
\section*{Acknowledgments}
We thank Fabian Jänsch, Radoslaw Adamczak, Saeid Haghighatshoar and
Giuseppe Caire for fruitful discussions.  PJ has been supported by DFG grant JU 2795/3.
\printbibliography

\end{document}